\documentclass[sn-mathphys-num]{sn-jnl}

\usepackage{graphicx}%
\usepackage{multirow}%
\usepackage{amsmath,amssymb,amsfonts, mathtools, commath}%
\usepackage{amsthm}%
\usepackage{mathrsfs}%
\usepackage[title]{appendix}%
\usepackage{xcolor}%
\usepackage{textcomp}%
\usepackage{manyfoot}%
\usepackage{booktabs}%
\usepackage{algorithm}%
\usepackage{algorithmicx}%
\usepackage{algpseudocode}%
\usepackage{listings}%
\usepackage{caption, subcaption}
\usepackage{float}
\usepackage{mathdots, nicematrix}
\usepackage{fontawesome5}
\usepackage{natbib}
\usepackage{epstopdf}

\newtheorem{theorem}{Theorem}[section]

\newtheorem{definition}{Definition}[section]
\newtheorem{lemma}{Lemma}[section]

\newcommand*{\C}{\mathbb{C}}
\newcommand*{\Q}{\mathbb{Q}}
\newcommand*{\R}{\mathbb{R}}

\newcommand*{\Z}{\mathbb{Z}}
\newcommand*{\F}{\mathbb{F}}

\begin{document}


\title[Fast Multiplication and the PLWE--RLWE Equivalence for an Infinite Family of Maximal Real Subfields of Cyclotomic Fields]{Fast Multiplication and the PLWE--RLWE Equivalence for an Infinite Family of Maximal Real Subfields of Cyclotomic Fields}




\author[1]{\fnm{Joonas} \sur{Ahola}}\email{joonas.ahola@aalto.fi}

\author[2]{\fnm{Iv\'an} \sur{Blanco-Chac\'on}}\email{ivan.blancoc@uah.es}

\author[3]{\fnm{Wilmar} \sur{Bolaños}}\email{\\\faEnvelope[regular] wilmar.bolanos@aalto.fi}
\author[4]{\fnm{Antti} \sur{Haavikko}}\email{anttijhaavikko@gmail.com}
\author[5]{\fnm{Camilla} \sur{Hollanti}}\email{camilla.hollanti@aalto.fi}
\author[6]{\fnm{Rodrigo M.} \sur{S\'anchez-Ledesma}}\email{rodrma01@ucm.es}


\affil[1]{, \orgname{Huawei Technologies Oy}, \orgaddress{ \city{Helsinki}, \country{Finland}}}

\affil[2,4]{\orgname{Universidad de Alcal\'a}, \orgaddress{ \city{Alcal\'a de Henares}, \country{Spain}}}

\affil[3,5]{\orgname{Aalto University}, \orgaddress{ \city{Espoo}, \country{Finland}}}



\affil[6]{ \orgname{University Complutense of Madrid / Indra}, \orgaddress{ \city{Madrid}, \country{Spain}}}


\abstract{We prove the equivalence between the Ring Learning With Errors (RLWE) and the Polynomial Learning With Errors (PLWE) problems for the maximal totally real subfield of the $2^r 3^s$-th cyclotomic field for $r \geq 3$ and $s \geq 1$. Moreover, we describe a fast algorithm for computing the product of two elements in the ring of integers of these subfields. This multiplication algorithm has quasilinear complexity in the dimension of the field, as it makes use of the fast Discrete Cosine Transform (DCT). Our approach assumes that the two input polynomials are given in a basis of Chebyshev-like polynomials, in contrast to the customary power basis. To validate this assumption, we prove that the change of basis from the power basis to the Chebyshev-like basis can be computed with $\mathcal{O}(n \log n)$ arithmetic operations, where $n$ is the problem dimension. Finally, we provide a heuristic and theoretical comparison of the vulnerability to some attacks for the $p$-th cyclotomic field versus the maximal totally real subextension of the $4p$-th cyclotomic field for a reasonable set of parameters of cryptographic size.}




\maketitle

\section*{Acknowledgements}
I. Blanco-Chacón, W. Bolaños, A. Haavikko, and C. Hollanti were supported in part by the Research Council of Finland (grant \#351271). C. Hollanti was also supported by the Finnish Ministry of Defence, MATINE (grant \#2500M-0147). I. Blanco-Chacón would like to acknowledge the visiting professor funding from the Aalto Science Institute and the Department of Mathematics and Systems Analysis at Aalto University, Finland. R. M. Sánchez-Ledesma was supported in part by the PQReact project, which has received funding from the European Union’s Horizon Europe research and innovation program under grant agreement \#101119547.

\section*{Statements and Declarations}
There are no competing financial nor non-financial interests that are directly or indirectly related to the work submitted for publication.

\section{Introduction}

Out of the first three post-quantum cryptography (PQC) standards published in 2024 by the National Institute of Standards and Technologies (NIST), two of them, namely ML-KEM \cite{FIPS:2024:203} and ML-DSA \cite{FIPS:2024:204}, belong to the \textit{lattice-based} family. This is not surprising, as this family has been dominant throughout all the stages of the standardization process: lattice-based cryptography has been the paradigm which has 
had the most schemes based upon, up until the final rounds. This is due to the fact that lattice-based post-quantum primitives can provide both acceptable key sizes and efficient performance while maintaining a solid understanding of the underlying mathematical hardness. The fact that lattice-based schemes enjoy homomorphic properties and can be turned into Fully Homomorphic Encryption (FHE) primitives is another important reason for the prevalence of the lattice family.

\subsection{LWE and its structured variants}

Within lattice-based cryptography, there exists one paradigm that has prevailed due to its simplicity, efficiency and versatility. It is referred to as the \textit{Learning With Errors (LWE)} paradigm. Intuitively, it amounts to the hardness of solving linear systems that have been perturbed by random noise. 

There exists a number of variants of this paradigm that use different mathematical structures for the data that forms the system. In purely LWE schemes, also referred to as \textit{unstructured} LWE variants, the mathematical structure employed is simply $\mathbb{Z}_q = \Z / q\Z$, the ring of rational integers modulo a rational prime $q$. 

The seek for significantly smaller key sizes for the sake of practical deployment and feasibility brought about \textit{structured} LWE variants, most notably, \textit{Ring LWE} and \textit{Polynomial LWE}. In these paradigms, the mathematical structure from which the terms are drawn is a more complex ring. For RLWE, the underlying ring is $R_q = \mathcal{O}_K / q\mathcal{O}_K$, where $\mathcal{O}_K$ is the ring of integers of a certain number field $K$. For PLWE we use $R_q=\mathcal{O}_f/q\mathcal{O}_f$, where $\mathcal{O}_f=\Z[x]/(f(x))$ and $f(x)\in \Z[x]$ is  monic and irreducible. 

To make our work self-contained it is convenient to recall the following definitions and facts:

\begin{definition}[The R/PLWE distributions]
Let $K$ be a number field, and $\mathcal{O}_K$ its ring of integers. Let $q$ be a rational prime, $f(x) \in \Z[x]$ a monic irreducible polynomial in $\Z[x]$, and $\mathcal{O}_f$ the associated quotient ring $\Z[x]/(f(x))$. Let $\chi$ be a discrete random distribution with values on $\mathcal{O}_K/q\mathcal{O}_K$ 
(resp. $\mathcal{O}_f/q\mathcal{O}_f$). For $s \in$ 
$\mathcal{O}_K/q\mathcal{O}_K$ (resp. $\mathcal{O}_f/q\mathcal{O}_f$), we define the (primal) RLWE (resp. PLWE) distribution $\mathcal{A}_{s, \chi}$ 
(resp. $\mathcal{B}_{s, \chi}$) as the distribution over 
$\mathcal{O}_K/q\mathcal{O}_K \times \mathcal{O}_K/q\mathcal{O}_K$ 
(resp. $\mathcal{O}_f/q\mathcal{O}_f \times \mathcal{O}_f/q\mathcal{O}_f$) obtained by sampling an element $a$ uniformly from $\mathcal{O}_K/q\mathcal{O}_K$ (resp. $\mathcal{O}_f/q\mathcal{O}_f$), $e$ drawn according to $\chi$, and outputting the pair $(a, a\cdot s + e)$.
\end{definition}

Associated with these distributions, the RLWE/PLWE search and decision problems are defined as follows:
\begin{definition}[R/PLWE problems]
Following the same notation as above, we have:

Search RLWE (resp. PLWE) asks an adversary to return the secret $s$ with non-negligible probability when this adversary is given access to arbitrarily many samples of the RLWE (resp. PLWE) distribution.

Decision RLWE (resp. PLWE) asks the adversary to decide whether a given random distribution is either uniform or the RLWE (resp. PLWE) distribution, with non-negligible probability when the adversary is given access to arbitrarily many samples of that given random distribution. 
\end{definition}

In literature, the most common choice is to use cyclotomic extensions $K = \Q(\zeta_{2n})$ with minimal polynomials $f(x) = x^n + 1$, where $n$ is a power of two or similarly the polynomial quotient rings $\mathcal{O}_f = \Z[x] / (x^n + 1)$. In fact, for this choice of $K$ and $f(x)$, the two rings $\mathcal{O}_K$ and $\mathcal{O}_f$ are isomorphic.

Note that for Galois extensions, the search problem can be effectively reduced to the decision problem \cite{EHL}. An even stronger result exists \cite{PRSD:2017:GRLWE}, providing a reduction for any ring and modulus.
Let us also mention that, for the RLWE distribution, we have given the definition of the \textit{primal} version. Another definition can be 
considered, by taking the secret $s$ and the support of $\chi$ over the trace-dual of the ring of integers $\mathcal{O}_{K}^{\vee}$. This is indeed the original definition of the RLWE problem in 
\cite{LPR:2013:ILL}. However, both problems were shown to be equivalent in \cite{RSW:2018:RPP}, and hence we will stick to the primal version for the sake of simplicity.

\hyphenation{par-a-digm}

Another relevant structured variant, as it is the paradigm used in the aforementioned standardized lattice-based schemes, is the \textit{Module 
LWE}. This paradigm, which loosely consists on working over $R_{q}^d$, is of great practical importance. The inclusion of the dimension $d$ allows a straightforward way of tuning the overall security of a scheme by simply modifying the parameter $d$, the dimension of the space. This modification has no practical impact on the mathematical object $R_q$ in which the operations take place. Therefore Module LWE provides a way to increase the security of the scheme without modifying the structure of $R_q$. 

It is worth noting that this cryptographic family has been found extremely useful not only in PQC but also in other relevant aspects of cryptography, most notably in \textit{Homomorphic Encryption} (HE). This type of cryptography inherits its name from the mathematical notion of \textit{homomorphisms}: intuitively, it allows to perform certain operations on the encrypted plaintext, in the same way as if we decrypt the ciphertext and apply the operation over the plaintext, mimicking the properties of homomorphic maps. The impact of practical schemes that achieve this property cannot be overstated, especially in fields in which the privacy of data is of the utmost importance. Depending on the nature of the operations allowed, we can distinguish between \textit{partially homomorphic encryption} or 
\textit{fully homomorphic encryption}, among others. The first provably lattice-based homomorphic encryption scheme is the GPV scheme 
\cite{10.1145/1374376.1374407}. Thereafter, a vast number of HE schemes have been derived from the LWE paradigm and its structured variants 
\cite{10.1145/2090236.2090262,BFV:2012:144, 10.1007/978-3-319-70694-8_15, 10.1007/978-3-642-40041-4_5, 10.1007/s00145-019-09319-x, 10.1007/978-3-662-46800-5_24}.

The decision whether to choose an unstructured or a structured paradigm is of great importance regarding the practical implementation and use: for purely unstructured LWE schemes, the size of the cryptographic keys needed for key exchange and digital signature schemes is $\mathcal{O}(n^2)$, where $n$ is a security parameter that represents the desired strength of the cryptosystem. Note that, for cryptographically relevant instances, this can mean keys of tens of thousands of bytes, a steep increase, even in the PQC entourage. Meanwhile, the additional structure added by the RLWE and PLWE variants decreases the size of 
cryptographic keys to $\mathcal{O}(n)$. While this value, in the order of thousands of bytes, represents a significant increase from pre-quantum 
primitives, it is well within the accepted price to pay within PQC.

Regarding PLWE, another highly relevant feature is achieved from the multiplicative structure. The multiplication of elements within the ring, needed for a number of cryptographic subprocesses of key generation, encryption and digital signatures, can be made very efficient through a number of optimizations, such as the Karatsuba or the Toom--Cook algorithms. For cyclotomic fields, the situation is even more optimal. We can use the more efficient Number Theoretic Transform (NTT), a special case of the Discrete Fourier Transform over finite fields, to perform the multiplications in quasilinear order. While this algorithm can be defined in the general case, its performance reaches an optimal point over \textit{cyclic} and \textit{nega-cyclic} fields.

The structure added in these variants gives us the benefit of reductions to approximate lattice problems like the shortest vector problem (SVP) over ideal lattices. However, the added algebraic structure can be a source of additional attacks that are ineffective against the unstructured variants. Particularly for PLWE, there exists a number of relevant attacks which employ information about the roots of the polynomial $f(x)$ behind the PLWE instance. In \cite{ELOS:2015:PWI,ELOS:2016:RCN}, a number of attacks against the decisional (and sometimes the search) version of PLWE are presented, if the polynomial $f(x)$ has a root $\alpha$ of \textit{small} order or of \textit{small} residue in $\mathbb{F}_q$, by constructing a number of distinguishability sets made from the evaluation of the samples on the critical root. Moreover, a more subtle family of statistical attacks can be deployed under certain conditions discussed therein. More recently, in \cite{BDM:2024:GARB}, these attacks (those based on roots of small order/residue as well as the statistical ones) have been generalized to polynomials with roots belonging to arbitrary degree field extensions of $\mathbb{F}_q$.

\subsection{Relation between RLWE and PLWE problems}
The RLWE and PLWE paradigms were introduced roughly at the same time. In 2009, \cite{SSTX:2009:EPK} presented a structured variant of the LWE problem based on ideal lattices over rings of the form $\mathcal{O} = \Z[x]/(f(x))$. They focused on cyclotomic polynomials giving rise to the first technical description of the PLWE paradigm.

Meanwhile, in 2010, \cite{LPR:2013:ILL} introduced another structured variant of LWE, this time over the ring of integers of algebraic number fields. While most of the results in this work were generalized to arbitrary number fields, the complete reduction involving the pseudorandomness of the RLWE distribution, an important feat with regards to practical cryptographic instances, 
was only proved for cyclotomic fields. The complete arbitrary number field generalization that lifted the requirement of cyclotomic fields was proved in \cite{PRSD:2017:GRLWE}. 

As noted above, both approaches have important cryptographic qualities: the RLWE problem bases its security on a more established and studied problem and has not been prone to many attacks over its structure. This is contrary to existing approaches towards exploiting the information given in the PLWE problem, such as root-based approaches 
\cite{ELOS:2015:PWI,ELOS:2016:RCN,BDM:2024:GARB}.

On the other hand, PLWE-based schemes enjoy faster cryptographic operations than RLWE-based ones, propelled by the use of the aforementioned NTT, Toom--Cook or Karatsuba algorithms to perform polynomial multiplications. This optimization technique is specially efficient over cyclotomic fields. Therefore, it is natural to ask whether these two important practical considerations can be present simultaneously, i.e., is it possible to have a scheme whose security is based on the RLWE problem, but enjoys the performance of PLWE-based schemes. A natural way to achieve this result is to study the equivalence between these two problems. We define computational equivalence as follows:

\begin{definition}[Equivalence between problems]
Two computational problems $A$ and $B$ are said to be equivalent if an oracle which solves one of the problems can be turned, in probabilistic polynomial time, into an oracle that solves the other and in such a way that the incurred noise increase is also polynomial. 
\end{definition}


The existence of equivalence between RLWE and PLWE for different families of number fields has been an important object of study. In \cite[Section 4]{RSW:2018:RPP}, this equivalence was proved for two reduced but large families of polynomials. For cyclotomic polynomials, \cite{SSS:2024:RPNE} showed that when no restriction is imposed on the conductor of the cyclotomic field, the two problems are not equivalent in general. 
However, if certain conditions are imposed on the conductor, it can be shown that these problems are indeed equivalent. In \cite{Blanco:2020:REC}, the author shows that if the conductor is divisible by a bounded number of primes, the condition number of the associated powers-of-the-roots Vandermonde matrix which reigns over the equivalence (which represents the variation induced from the difference between the coordinate and canonical embedding associated with the PLWE and RLWE problems, respectively) is bounded by a polynomial factor and, therefore, the problems are equivalent.

Another topic which has recently gained traction is the study of the R/PLWE problem for the maximal totally real subextension of a cyclotomic field. The interest for this field was introduced in \cite{BL:2022:RPE,Blanco:2021:VMTR}, where it 
was proved that the defining polynomials for these subfields are resistant to PLWE root-based attacks for $\alpha = \pm 2$ as well as for $\alpha=\pm 1$. 

Furthermore, it is a strategic approach to have at hand as large a family of fields and rings as possible, with the idea of being prepared for a potential catastrophic attack against current structured cryptosystems. Indeed, it is worth to mention \cite{cramerducaswesolowsky}, where the authors exploit certain class-field-theoretical properties of the Stickelberger ideal to prove the following fact:

\begin{theorem}Assuming the Generalised Riemann Hypothesis, there exists a quantum polynomial time algorithm, that given an ideal $\frak{a}$ of $\mathcal{O}_K$ for $K$ a cyclotomic number field of prime power conductor, returns an element $v\in \frak{a}$ of Euclidean norm 
$$
||v||\leq N\frak{a}^{1/n}\exp(\mathcal{O}(\sqrt{n})).
$$
\label{CDW}
\end{theorem}

Even more striking is the recent refinement of Theorem \ref{CDW}, which shows that under mild but technical conditions, the quantum algorithm runs in probabilistic polynomial time in $[K:\mathbb{Q}]$, $h_K^{+}$ and $\log(N\frak{a})$, where $h_K^{+}$ is the narrow class number of $K$. The reader is referred to \cite[Theorem 5.1]{CDW2} for further details.

Regardless of the theoretical nature of these results, they suggest that the ease and amenable arithmetical features that cyclotomic fields enjoy might give rise to vulnerabilities or at least theoretical concerns. It is also noteworthy that during the last three years, several research teams have published a number of attacks against cyclotomic-based RLWE schemes. Even though some results have turned out to be incorrect and consequently withdrawn, this shows that the cryptanalysis of cyclotomic-based primitives is an active research area, and the use of more subtle number-theoretical and class-field-theoretical ideas might yield critical improvements in the existing cryptanalytic algorithms.


\subsection{Contributions and Organization}


In this work, we continue the mathematical analysis of maximal totally real subextensions of cyclotomic fields. In particular, we focus in the case of conductor $n = 2^r3^s$ with $r>1$ and $s\geq 1$. Notice that for this family of fields the hypothesis of Theorem \ref{CDW} does not apply, i.e., the ring of integers of this maximal totally real number field is not an ideal of the ring of integers of the cyclotomic field. Moreover, the conductor can be chosen not to be a prime power. Hence, these extensions are not under the theoretical threat provided by that result. Beyond that, the main contributions of our work are:

First, we prove that the RLWE and PLWE problems are equivalent for this family of fields. Second, we demonstrate that these rings exhibit fast quasilinear multiplication. Finally, a numerical study is carried out related to small roots of cyclotomic polynomials and the ones corresponding to the maximal real subextensions of matching degree. It will be shown that heuristically, the polynomials generating the real subextensions may be less likely to have other \textit{weak} small roots $\alpha \in S = \{\pm 2, \pm 3, \pm 4, \pm 8\}$ than cyclotomic polynomials, a fact that sparks further interest in their study to back PLWE schemes. We constructed the set $S$ to consist of small powers of 2 and the prime 3 in order to witness if primes and powers of primes exhibit similar behaviour as roots. Conclusions are drawn in different parameter regimes, showing that the maximal real subextensions provide more frequent secure instances for some parameters. Regardless of these promising findings, this work should be regarded as a mathematical study rather than a suggestion for a new PLWE setting. 

Previous works related to ours include \cite{Blanco:2021:VMTR} and \cite{BL:2022:RPE} addressing the RLWE--PLWE equivalence for maximal totally real cyclotomic subextensions, and \cite{lyubashevsky2019nttru}, where an NTT-based multiplication scheme is proposed for the non-power-of-two $2^r3^s$-th cyclotomic fields, with $r,s\geq 1$.


This work is organized as follows: Section \ref{sec:DCT} introduces the Discrete Cosine Transform ($\mathsf{DCT}$), a variation of the Discrete Fourier Transform (DFT). The DCT will be of importance to achieve efficient polynomial multiplication over the totally maximal real subfields. In Section \ref{sec:DCT}, we also recall the notion of condition number, which is the formal tool to capture the notion of the RLWE--PLWE equivalence. We exploit basic properties of the $\mathsf{DCT}$ to give an upper bound for the condition number of a certain matrix that is used in the proof of the RLWE--PLWE equivalence in the next section.

Section \ref{sec:RLWE-PLWE_equivalence} provides a proof of the RLWE--PLWE equivalence for the totally maximal real subextension of the cyclotomic field $\Q(\zeta_n)$, where $\zeta_n$ is a primitive $n$-th root of unity in $\C$ and $n = 2^r3^s$. We make use of a basis of modified Chebyshev polynomials instead of the usual power basis to polynomially bound the condition number of the matrix associated to the canonical embedding.

Section \ref{sec:fast_multiplication_via_DCT} shows that by means of the $\mathsf{DCT}$ efficient polynomial multiplication can be carried out in the modified Chebyshev basis with complexity of order $\mathcal{O}(m \log m)$ on the conductor $m$. 

Section \ref{sec:PLWE_in_maxreal} makes the base change computations explicit by again using the $\mathsf{DCT}$ and Chebyshev interpolation. The efficiency of the base change is also quasilinear, meaning $\mathcal{O}(m \log m)$, where $m$ is the degree of the polynomials in the PLWE scheme. Therefore, we can conclude that for the totally maximal real subextension of a cyclotomic field of conductor $m = 2^r3^s$, we can simultaneously inherit the security warranty from the associated number field and RLWE problem to the PLWE problem, while maintaining asymptotically fast polynomial multiplication. 

Finally, Section \ref{sec:root_computations} gives heuristic evidence for the resistance of the maximal totally real cyclotomic polynomials against root-based attacks. Our numerical computations show that in certain parameter regimes the minimal polynomials of $\zeta_n+\zeta_n^{-1}$ may be less prone to having small roots than their cyclotomic counterparts, therefore making them less likely to be vulnerable to root-based attacks for the PLWE paradigm.

\section{The Discrete Cosine Transform, $\mathsf{DCT}$}
\label{sec:DCT}
The Discrete Cosine Transform ($\mathsf{DCT}$) is widely used in many
digital signal processing applications. Multiple fast algorithms for different types of the Discrete Cosine Transform have been reported in the literature \cite{bi1998dct, hou1987fast, kok1997fast, ahmed1974discrete}. In this paper, we shall use types II and III, which are mutual inverses when scaled properly.

\begin{definition}[$\mathsf{DCT}$]
Let $N \in \Z^+$ and $a(k), \, k= 0,1,\dots, N-1$ a real sequence. The non-scaled type-III Discrete Cosine Transform of $a(k)$ is the sequence $\hat{a}(j)$ defined by 
$$\hat{a}(j) = \frac{a(0)}{2} + \sum_{i=1}^{N-1} a(i)\cos\left( \frac{2\pi(2j +1) i}{4N}\right),  \quad 0\leq j \leq N-1.$$

The inverse of the type-III $\mathsf{DCT}$ is given by the type-II $\mathsf{DCT}$. The non-scaled type-II Discrete Cosine Transform of the sequence $a(k)$ is a new sequence $a'(j)$ given by
$$ a'(j) =  \sum_{i=0}^{N-1} a(i)\cos \left( \frac{2\pi(2i+1)j}{4N} \right), \quad 0\leq j \leq N-1.$$
\end{definition}

In order to simplify the notation, we introduce an $N \times N$ cosine matrix $\mathbf{C}_N$ and a diagonal scale matrix $\mathbf{S}_N$ defined by 
\begin{align*}
    \left(\mathbf{C}_N \right)_{ij} &= \cos \left(\frac{2\pi(2i +1)j}{4N} \right) \quad \text{for } i,j=0,1,\dots,N-1,
\end{align*}
and
\begin{align*}
     \mathbf{S}_N &= 
     \begin{bmatrix}
         \underset{1 \times 1 }{2} & \underset{1 \times (N-1)}{\mathbf{0}} \\
         \underset{(N-1) \times 1}{\mathbf{0}} & \underset{(N-1) \times (N-1)}{\mathbf{I}}
     \end{bmatrix}.
\end{align*}
In this way, we can write the two transforms in matrix notation:
\begin{equation}
\mathsf{DCT}(\mathbf{a}) := \mathbf{C}_N \mathbf{S}_N^{-1}\mathbf{a},    
\end{equation}
and the inverse transform
\begin{equation} 
\mathsf{IDCT}(\mathbf{a}) := \mathbf{C}_N^T \mathbf{a}. 
\end{equation}

As mentioned before, the scaled type-III $\mathsf{DCT}$ and scaled type-II $\mathsf{DCT}$ are inverses of each other. For the non-scaled versions we have a similar result:

\begin{lemma}\label{lemma dct} For any real sequence $a(k)$, $k=0,1,\dots, N-1$, we have 
$$\mathsf{IDCT} \left( \mathsf{DCT}(\mathbf{a}) \right) = \frac{N}{2}\mathbf{a}.$$
In matrix notation,
\begin{equation}\label{cosine transform}
\mathbf{C}_N^T \mathbf{C}_N \mathbf{S}_N^{-1} = \frac{N}{2} \mathbf{I}_N
\end{equation}
or equivalently 
$$\mathbf{C}_N^T \mathbf{C}_N  = \frac{N}{2} \mathbf{S}_N.$$
\end{lemma}

\begin{proof}
First, let us deduce the following orthogonality relation 
\begin{equation}
\sum\limits_{k=0}^{N-1} \cos \left( \frac{2\pi(2k +1)j}{4N} \right) = 0.
\end{equation}
Let $j=1,2,\dots, 2N-1$ and $\theta = \frac{\pi j}{N}$. Then
\begin{align*}
  2\sin\left(\frac{\theta}{2}\right) \sum_{k=0}^{N-1}\cos\left(\frac{2\pi(2k + 1)j}{4N}\right) & = 2\sin\left(\frac{\theta}{2}\right) \sum_{k=0}^{N-1}\cos\left(\frac{\pi k j}{N} + \frac{\pi j}{2N}\right) \\
  &= \sum_{k=0}^{N-1} 2\cos\left(\theta k + \frac{\theta}{2}\right) \sin\left(\frac{\theta}{2}\right)  \\
  &=  \sum_{i=0}^{N-1} \sin(\theta k + \theta) - \sin(\theta k)\\
  &= \sin(N \theta) - \sin(0)\\
  &= 0.
\end{align*}
     
Next, we compute the diagonal entries $(\mathbf{C}_N^T \mathbf{C}_N)_{jj}.$ For $j = 0$, it is clear that $(\mathbf{C}_N^T \mathbf{C}_N)_{jj} = N$. Otherwise for $j = 1, \hdots, N-1$, we find that
\begin{align*}
    (\mathbf{C}_N^T \mathbf{C}_N)_{jj} &=  \sum_{k=0}^{N-1} \cos^2\left(\frac{2\pi(2k + 1)j}{4N}\right) \\
    &=  \frac{1}{2}\sum_{k=0}^{N-1} \left(1 + \cos\left(\frac{2 \cdot 2\pi(2k + 1)j}{4N}\right) \right ) \\
    &= \frac{N}{2} + \frac{1}{2} \sum_{k=0}^{N-1}\cos\left(\frac{2\pi(2k + 1)j}{2N}\right) = \frac{N}{2}.
\end{align*}
    
Finally, for $i\neq j$ it holds that
\begin{align*}
    (\mathbf{C}_N^T \mathbf{C}_N)_{ij} &=  \sum_{k=0}^{N-1} \cos\left(\frac{2\pi(2k + 1)i}{4N}\right) \cos\left(\frac{2\pi(2k + 1)j}{4N}\right)   \\
    &=  \frac{1}{2}\sum_{k=0}^{N-1} \cos\left(\frac{ 2\pi(2k + 1)(i+j)}{4N}\right ) + \cos\left(\frac{ 2\pi(2k + 1)(i-j)}{4N}\right ) 
    \\
    &= 0.
\end{align*}
\end{proof}
Next, we recall the definition of the condition number of a matrix, a tool with which the notion of equivalence is stated.

\begin{definition}Let $A\in\mathrm{GL}_n(\mathbb{C})$ be an invertible square matrix with complex entries. The condition number is defined as
$$
\kappa_F(A)=||A||_F||A^{-1}||_F,
$$
where $||\phantom{A}||_F$ denotes the Frobenius norm.
\end{definition}
Now, we can prove the following:
\begin{lemma}\label{lemma cond}
    The condition number of the cosine matrix  $\mathbf{C}_N$ is
    \begin{align*}
        \kappa_F(\mathbf{C}_N)^2 &:= ||\mathbf{C}_N ||_F^2||\mathbf{C}^{-1}_N ||_F^2 =  N^2 + \frac{N-1}{2} = \mathcal{O}(N^2).
    \end{align*}
    In particular, $||\mathbf{C}_N ||_F^2 = N + N(N-1) / 2$ and $||\mathbf{C}_N^{-1} ||_F^2 = (2N - 1) / N$.
\end{lemma}

\begin{proof}
    We know that $||\mathbf{C}_N||^2 = \operatorname{Tr}(\mathbf{C}_N^T \mathbf{C}_N)$. Therefore, from the previous Lemma \ref{lemma dct}, we have $\mathbf{C}_N^T \mathbf{C}_N  = \frac{N}{2} \mathbf{S}_N,$ and $|| \mathbf{C}_N||^2 = N + N(N-1)/2$. 
    
    On the other hand, by using \eqref{cosine transform} the inverse matrix $\mathbf{C}_N^{-1}$ can be written as
$$ \mathbf{C}_N^{-1} = \frac{2}{N}\mathbf{S}_N^{-1}\mathbf{C}_N^{T}.$$    
It follows that
$$\mathbf{C}_N^{-1}\mathbf{C}_N^{-T} = \frac{2}{N}\mathbf{S}_N^{-1}.$$
Now, we have everything we need to compute the Frobenius norm of the inverse transform,
\begin{align*}
    || \mathbf{C}^{-1}_N||_{F}^2 & = \operatorname{Tr}(\mathbf{C}_N^{-T} \mathbf{C}_N^{-1}) \\
    & = \operatorname{Tr}(\mathbf{C}_N^{-1} \mathbf{C}_N^{-T}) \\
    & = \operatorname{Tr} \left( \frac{2}{N}\mathbf{S}_N ^{-1}\right) \\
    & = \frac{2}{N}\left( \frac{1}{2} + N-1 \right) = \frac{2N-1}{N}.
\end{align*}
Finally, the squared condition number of the cosine matrix $\mathbf{C}_N$ is
    \begin{align*}
        \kappa_F(\mathbf{C}_N)^2 &= ||\mathbf{C}_N||_F^2||\mathbf{C}_N^{-1}||_F^2 = \frac{N(N+1)}{2}\cdot \frac{2N -1}{N} =  N^2 + \frac{N-1}{2}.
    \end{align*}
\end{proof}

\section{RLWE--PLWE equivalence for $n = 2^r 3^s$}
\label{sec:RLWE-PLWE_equivalence}

The study of the equivalence between RLWE and PLWE was first addressed in \cite{DD:2017} for cyclotomic number fields of degree $2^kp$ or $2^kpq$ with $p,q$ primes and $q<p$, and later on in \cite[Section 4]{RSW:2018:RPP} for two particular non-cyclotomic families of polynomials. In \cite{DSS:2022}, the authors obtained an explicit closed formula for the condition number of the $2^kp^l$ cyclotomic field, and in \cite{SSS:2024:RPNE} the authors showed that for an arbitrary conductor, the two problems are not equivalent for cyclotomic fields.

Still within the cyclotomic family, in \cite{Blanco:2020:REC}, it was proved that if the conductor is divisible by a bounded number of primes, the RLWE--PLWE equivalence holds alongside with much sharper bounds for the condition number for conductors divisible by up to $3$ primes. This result was recently generalized in \cite{BPBN} to conductors divisible by up to $6$ primes. The authors also introduce and study cyclo-multiquadratic number fields, a family which also enjoys RLWE--PLWE equivalence and fast multiplication.

Beyond cyclotomics, in \cite{Blanco:2021:VMTR}, the RLWE--PLWE equivalence was proved for the maximal totally real subfield of the $4p$-th cyclotomic field. This result was generalized in \cite{BL:2022:RPE} to conductors of the form $2^rpq$ with $p<q$ primes, as well as for $p=1$ or $q=1$. Furthermore, in this last work it was proved that the defining polynomials of these subfields are resistant to the PLWE root-based attacks for $\alpha = \pm 1, \pm 2$. This gives a reason to pay attention to these fields from a cryptographic point of view. However, cyclotomic polynomials can vanish at $\alpha=2$ modulo suitable primes, and the authors provided some examples.

To set the stage for our result, let $\psi_n = \zeta_n + \zeta_n^{-1} = 2\cos(2\pi / n)$ so $\Q(\psi_n) = \Q(\zeta_n)^+$ is the maximal real subfield of the cyclotomic field $\Q(\zeta_n)$. We know that $[\Q(\psi_n) : \Q] = \phi(n) / 2$, where $\phi$ is Euler's totient function. It is well known that the ring of integers of $\Q(\psi_n)$ is $\Z[\psi_n] \cong \Z[x] / (\Psi_n(x))$, where $\Psi_n(x)$ is the minimal polynomial of $\psi_n$ of degree $\phi(n) / 2$. For these standard claims, we refer to \cite{washington2012introduction}.

Chebyshev polynomials exhibit nice properties with the cosine function, especially the cosine of a multiple of an angle. Since the element of interest $\psi_n = 2\cos(2\pi / n)$ has an extra factor of 2, we introduce a modified family of Chebyshev polynomials.
\begin{definition}
    Let $T_n(x)$ be the Chebyshev polynomial of the first kind of degree $n$. We define the family of polynomials $V_n(x)$ with $V_0(x)= 1$ and 
    \begin{align*}
        V_i(x) = 2T_i(x/2) \quad \text{for } i \geq 1
    \end{align*}
     to be the modified Chebyshev polynomial of degree $i$.
    \label{def:modified_chebyshev_polynomials_V}
\end{definition}
The family of modified Chebyshev polynomials satisfy the property
\begin{align*}
    V_n(2\cos(\theta)) = 2\cos(n\theta) \quad \text{for any } \theta  \text{ and } n \geq 1.
\end{align*}
Alternatively, $V_n(x)$ can be defined by the recursion 
$$V_n(x) = x V_{n-1}(x) - V_{n-2}(x) \mbox{ \; for \;} n \geq 3,$$
with the initializing sequence
$$V_0(x)=1 ,\; V_1(x) = x , \; V_2(x) = x^2 - 2.$$

Fix $n = 2^r 3^s$ and $m = \phi(n)/2 = 2^{r-1}3^{s-1}$. Then $$V := \{V_0(x), V_1(x), \hdots, V_{m-1}(x)\}$$ is basis for $\mathfrak{O} = \Z[x] / (\Psi_n(x))$ as all the polynomials $V_i(x)$ are monic with degree $i$. The Minkowski canonical embedding $\mathcal{M}: \Z[x] / (\Psi_n(x)) \to \R^m$ is given by
\begin{align*}
    a_0V_0(x) + a_1V_1(x) + \hdots + a_{m-1}V_{m-1}(x) &\mapsto \mathbf{M} (a_0, \, a_1, \hdots, \ a_{m-1})^T,
\end{align*}
where
\begin{align*}
    \underset{m \times m}{\mathbf{M}} &=
    \begin{bmatrix}
        1 & 2\cos(2\pi / n) & 2 \cos(2 \pi 2 / n) & \hdots & 2\cos(2\pi (m-1) /n) \\
        \vdots & \vdots & \vdots & \hdots  & \vdots \\
        1 & 2\cos(2\pi \sigma / n) & 2 \cos(2 \pi \sigma 2 / n) & \hdots & 2\cos(2\pi \sigma (m-1) /n) \\
        \vdots & \vdots & \vdots & \ddots & \vdots \\
        1  & \hdots & \hdots & \hdots & 
    \end{bmatrix}
\end{align*}
is an $m$-by-$m$ matrix with $\sigma \in \{1,2,\hdots, n/2\}$ and $(\sigma, n) = 1$.

To show that the PLWE and RLWE problems are equivalent, it is enough to prove that the condition number of the matrix $\mathbf{M}$ is bounded by a polynomial in $n$. We adapt this approach from \cite{DD:2017, RSW:2018:RPP}. For a short summary, see \cite{Blanco:2021:VMTR} and the discussion following Definition 2.3. If we can find a polynomial bound in $n$, we say that $\mathbf{M}$ is well conditioned. To see that $\mathbf{M}$ is well conditioned, first note that the matrix $\mathbf{M}$ is well conditioned if and only if the matrix
\begin{align*}
    \underset{m \times m}{\mathbf{V}} &=
    \begin{bmatrix}
        1 & \cos(2\pi / n) & \cos(2 \pi 2 / n) & \hdots & \cos(2\pi (m-1) /n) \\
        \vdots & \vdots & \vdots & \hdots  & \vdots \\
        1 & \cos(2\pi \sigma / n) & \cos(2 \pi \sigma 2 / n) & \hdots & \cos(2\pi \sigma (m-1) /n) \\
        \vdots & \vdots & \vdots & \ddots & \vdots \\
        1  & \hdots & \hdots & \hdots & 
    \end{bmatrix}
\end{align*}
is well conditioned as it is a simple scaling of the columns of $\mathbf{M}$. In fact, $\mathbf{M} = 2 \mathbf{V} \mathbf{S}_m^{-1}$. We are ready to state our first main result.

\begin{theorem}
    Let $r\geq 3$, $s \geq 1$, and $n = 2^r 3^s$. Then PLWE and RLWE are equivalent for the maximal real subextension of the $n$-th cyclotomic field.
\end{theorem}

\begin{proof}
To match the denominators of the arguments of the cosines in $\mathbf{V}$ and $\mathbf{C}_N$, we set $n = 4N = 2^r 3^s$. This gives us the relation $N = 3m/2$ for the matrix dimensions. Notice that $\mathbf{V}$ is a $m \times m$ matrix, but $\mathbf{C}_N$ is a larger $N \times N$ matrix. 
Additionally, we can directly see that each entry of $\mathbf{V}$ is of the form $\cos\left( \frac{2\pi \sigma j}{2^r3^s}\right)$ and every entry of $\mathbf{C}_N$ is of the form $\cos\left(\frac{2\pi (2i +1) j}{2^r3^s} \right).$
Since $(\sigma,n)=1$ implies $\sigma$ is odd and $3 \nmid \sigma$, we conclude that $\mathbf{V}$ is actually a submatrix of $\mathbf{C}_N$ associated to the rows $i$ such that $3 \nmid (2i + 1)$. 

Thus, let $\mathbf{P}$ be the row permutation matrix associated to the permutation $p(i)$ of $\{0,1,2,\dots, N-1\}$ such that $3 \nmid 2p(i) + 1$ for $i=0,1,\dots,m-1$ and for the rest of the row indices $i=m,m+1, \dots, N-1$, $3 | 2p(i) +1$. Then the row permutation yields
\begin{align*}
    \mathbf{P} \mathbf{C}_N &=
    \begin{bmatrix*}
        \underset{m \times m}{\mathbf{V}} & \underset{m \times m/2}{\mathbf{B}} \\ \underset{m/2 \times m}{\mathbf{A}} & \underset{m/2 \times m/2}{\mathbf{C}}
    \end{bmatrix*},
\end{align*}
where all the matrices have entries of the form $\cos( 2\pi \sigma j / n)$
with $\sigma$ odd. Furthermore, we have the additional conditions that 
\begin{align*}
    &\text{for } \mathbf{A}, \ j = 0, \hdots, m-1 \text{ and } 3 \mid \sigma, \\
    &\text{for } \mathbf{B}, \ j = m, \hdots, N-1 \text{ and } 3 \nmid \sigma, \\
    &\text{for } \mathbf{C}, \ j = m, \hdots, N-1 \text{ and } 3 \mid \sigma.
\end{align*}

Finally, we observe that for all column indices $j > m$, we have a unique index $j' \in \{1, \hdots, m/2\}$ such that $j = m + j'$. Also, let $\tilde{j} = m - j'$. Then
\begin{align*}
    \cos\left(\frac{2 \pi \sigma j}{n} \right) + \cos\left(\frac{2 \pi \sigma \tilde{j}}{n} \right) &= \cos\left(\frac{2 \pi \sigma m}{n} + \frac{2 \pi \sigma j'}{n} \right)
    + \cos\left(\frac{2 \pi \sigma m}{n} - \frac{2 \pi \sigma j'}{n} \right) \\
    &= 2 \cos\left(\frac{2 \pi \sigma m}{n} \right)\cos\left(\frac{2 \pi \sigma j'}{n} \right) \\
    &= 2 \cos\left(\frac{2 \pi \sigma}{6} \right)\cos\left(\frac{2 \pi \sigma j'}{n} \right).
\end{align*}
The first cosine term in the product can attain two values depending on the divisibility of $\sigma$. If $3 \nmid \sigma$, we know that $(\sigma, 6) = 1$. Hence, we can write $\sigma = 6l \pm 1$ for some $l \in \Z$. The other option is that $3 \mid \sigma$, and then $2 \sigma / 6$ is an odd integer. We separate the two cases as follows:
\begin{align*}
    2 \cos\left(\frac{2 \pi \sigma}{6} \right) &= 
    \begin{cases}
        2 \cdot  \cos\left(\pi (\sigma/3) \right), & \text{ if } 3 \mid \sigma \\
        2 \cdot  \cos\left(\frac{2 \pi (6l \pm 1)}{6}  \right),  & \text{ if } 3 \nmid \sigma
    \end{cases} \\
    &= \begin{cases}
        -2,  &  \text{ if }  3 \mid \sigma \\
        1, &  \text{ if }  3 \nmid \sigma.
    \end{cases}
\end{align*}

We want to use this simplification of the sum of the two cosine terms columnwise on the matrix $\mathbf{PC}_N$. To perform the summation as a column operation on $\mathbf{PC}_N$, let
\begin{align*}
    \underset{N \times N}{\mathbf{R}_1} &= 
    \begin{bmatrix*}
        \underset{m \times m}{\mathbf{I}} & \underset{m \times (m/2)}{\mathbf{D}} \\ 
        \underset{(m/2) \times m}{\mathbf{0}} & \underset{(m/2) \times (m/2)}{\mathbf{I}}
    \end{bmatrix*},
\end{align*}
where
\begin{align*}
    \underset{m \times (m/2)}{\mathbf{D}} &= 
    \begin{bmatrix*}
        \underset{(m/2+1) \times 1}{\mathbf{0}} & \underset{(m/2+1) \times (m/2 -1)}{\mathbf{0}} \\
        \underset{(m/2-1) \times 1}{\mathbf{0}} & \underset{(m/2-1) \times (m/2 -1)}{\mathbf{J}}
    \end{bmatrix*},
\end{align*}
and $\mathbf{J}$ denotes a row-reversed identity matrix.
The column operation $\mathbf{P}\mathbf{C}_N \mathbf{R}_1$ yields
\begin{align*}
    \mathbf{P}\mathbf{C}_N \mathbf{R}_1
    &= 
    \begin{bmatrix*}
        \underset{m \times m}{\mathbf{V}} & \underset{m \times (m/2)}{\mathbf{B'}} \\ \underset{(m/2) \times m}{\mathbf{A}} & \underset{(m/2) \times (m/2)}{\mathbf{C'}}
    \end{bmatrix*}
\end{align*}
where the first columns of $\mathbf{B'}$ and $\mathbf{C'}$ are vectors of lengths $m$ and $m/2$ consisting only of elements $1/2$ and $-1$, respectively. Otherwise, the rest of $\mathbf{B'}$ comprises elements of the form $\cos(2 \pi \sigma j' / n)$ with $3 \nmid \sigma$ and $j' = 1, \hdots, m/2-1$. Similarly, for $\mathbf{C'}$, the elements are of the form $-2\cos(2 \pi \sigma j' /n)$ for $j' = 1, \hdots, m/2-1$ but $3 \mid \sigma$. The coefficients multiplying the cosines are direct result of the simplification step described above.

In order to bound the norm of the block $\mathbf{V}$ in the matrix $\mathbf{PC}_N$, we want to bring the matrix $\mathbf{PC}_N$ to a block lower triangular form with only zero entries in the upper-right block. For another set of column operations, let
\begin{align*}
    \underset{N \times N}{\mathbf{R}_2} &= \begin{bmatrix*}
        \underset{m \times m}{\mathbf{I}} & \underset{m \times (m/2)}{\mathbf{E}} \\ \underset{(m/2) \times m}{\mathbf{0}} & \underset{(m/2) \times (m/2)}{\mathbf{I}}
    \end{bmatrix*},
\end{align*}
where
\begin{align*}
    \mathbf{E} &= 
    \left[
    \begin{array}{c | c}
        \underset{1 \times 1}{-1/2} & \underset{1 \times (m/2 -1)}{\mathbf{0}} \\
        \hline
        \underset{(m/2 -1) \times 1}{\mathbf{0}} & \underset{(m/2 -1) \times (m/2 -1)}{-\mathbf{I}} \\
        \hline
        \underset{(m/2) \times 1}{\mathbf{0}} & \underset{(m/2) \times (m/2 -1)}{\mathbf{0}}
    \end{array}
    \right].
\end{align*}
Then another set of column operations gives us the block triangular matrix
\begin{align*}
    \mathbf{P}\mathbf{C}_N \mathbf{R}_1 \mathbf{R}_2
    &= 
    \begin{bmatrix*}
        \underset{m \times m}{\mathbf{V}} & \underset{m \times (m/2)}{\mathbf{0}} \\ \underset{(m/2) \times m}{\mathbf{A}} & \underset{(m/2) \times (m/2)}{\mathbf{C''}}
    \end{bmatrix*},
\end{align*}
where $\mathbf{C''}$ consists of elements of the form $-3\cos(2 \pi \sigma j' /n)$ for $j' = 0, \hdots, m/2-1$ where $3 \mid \sigma$.

Finally, we combine the column operations into a single matrix
\begin{align*}
    \underset{N \times N}{\mathbf{R}} &= \mathbf{R}_1\mathbf{R}_2 =
    \begin{bmatrix*}
        \underset{m \times m}{\mathbf{I}} & \underset{m \times (m/2)}{\mathbf{F}} \\
        \underset{(m/2) \times m}{\mathbf{0}} & \underset{(m/2) \times (m/2)}{\mathbf{I}}
    \end{bmatrix*},
\end{align*}
where $\mathbf{F} = \mathbf{E} + \mathbf{D}$. Note that the nonzero entries of $\mathbf{E}$ and $\mathbf{D}$ do not overlap. Moreover, the inverse of the matrix $\mathbf{R}$ is
\begin{align*}
    \underset{N \times N}{\mathbf{R}^{-1}} &= 
    \begin{bmatrix*}
        \underset{m \times m}{\mathbf{I}} & \underset{m \times (m/2)}{-\mathbf{F}} \\
        \underset{(m/2) \times m}{\mathbf{0}} & \underset{(m/2) \times (m/2)}{\mathbf{I}}
    \end{bmatrix*}.
\end{align*}
Now we can bound the Frobenius norm of $\mathbf{V}$. From the first block form
\begin{align*}
    \mathbf{P} \mathbf{C}_N &=
    \begin{bmatrix*}
        \underset{m \times m}{\mathbf{V}} & \underset{m \times (m/2)}{\mathbf{B}} \\ \underset{(m/2) \times m}{\mathbf{A}} & \underset{(m/2) \times (m/2)}{\mathbf{C}}
    \end{bmatrix*}
\end{align*}
we observe that $\mathbf{V}$ is a submatrix, so we get the strict bound
\begin{align*}
    ||\mathbf{V}||_F^2 < || \mathbf{P}\mathbf{C}_N ||_F^2 = ||\mathbf{C}_N ||_F^2.
\end{align*}
To bound the norm of the inverse $\mathbf{V}^{-1}$, we use the column reduced form given by the right multiplication by $\mathbf{R}$,
\begin{align*}
    ||\mathbf{V}^{-1}||_F^2 &< || \mathbf{R}^{-1}\mathbf{C}_N^{-1} \mathbf{P}^T ||_F^2 \leq || \mathbf{R}^{-1} ||_F^2 ||\mathbf{C}_N^{-1} \mathbf{P}^T ||_F^2 \\
    &= || \mathbf{R}^{-1} ||_F^2 ||\mathbf{C}_N^{-1}||_F^2 = (N + ||\mathbf{F}||_F^2) ||\mathbf{C}_N^{-1}||_F^2  < 2N ||\mathbf{C}_N^{-1}||_F^2.
\end{align*}
By combining the two bounds we get a polynomial upper bound for the condition number of $\mathbf{V}$, 
\begin{align*}
    \kappa_F(\mathbf{V})^2 &= ||\mathbf{V}{-1}||_F^2 ||\mathbf{V}||_F^2
    < 2N ||\mathbf{C}_N^{-1}||_F^2 ||\mathbf{C}_N||_F^2 = 2N \kappa_F(\mathbf{C}_N)^2.
\end{align*}
As a last step, from Lemma \ref{lemma cond} we know that 
$\kappa_F(\mathbf{C}_N)^2 = \mathcal{O}(N^2)$. Putting everything together yields the desired bound 
\begin{align*}
    \kappa_F(\mathbf{V})^2 = 2N \mathcal{O}(N^2) = \mathcal{O}(N^3) = \mathcal{O}(n^3).
\end{align*}
This concludes the proof.
\end{proof}

\section{Fast multiplication in $\Z[x]/(\Psi_n(x))$ via the $\mathsf{DCT}$}
\label{sec:fast_multiplication_via_DCT}
In this section, we introduce an algorithm for fast multiplication over the ring $\Z[x]/(\Psi_n(x))$. The algorithm is similar to the Number Theoretic Transform (NTT). However, the NTT is typically computed over the ring $\Z[x]/(x^n + 1)$, where $n$ is a power of two yielding the ring of integers of the $2n$-th cyclotomic field \cite{agarwal1975ntt, pollard1971fast}. For non-power-of-two cyclotomic fields, in \cite{lyubashevsky2019nttru} the authors describe an NTT-based multiplication algorithm in the $2^r 3^s$-th cyclotomic field, where $r \geq 1$ and $s \geq 1$. The algorithms developed to compute the $\mathsf{DCT}$ and its inverse $\mathsf{IDCT}$ can be classified into two categories: indirect and direct. Indirect algorithms take advantage of existing fast algorithms, such as the fast Fourier and Hadamard transforms \cite{ahmed1974discrete}. On the other hand, direct algorithms reduce the computational complexity by means of matrix factorization and recursive decomposition as a generalization of the Cooley--Tukey algorithm with the same asymptotic complexity $\mathcal{O}(n \log n)$ \cite{bi1998dct, bi1999fast, hou1987fast, kok1997fast}.

Next, we show how the $\mathsf{DCT}$ and its inverse can be used for fast multiplication algorithms. If $p(x)$ is a polynomial of degree less than or equal to $N-1$, then $p(x)$ can be represented in base $\{V_0(x), V_1(x), \dots, V_{N-1}(x) \}$ as  
$$p(x) = \sum_{i = 0}^{N-1} a_iV_i(x).$$
We define $\mathsf{DCT}(p(x)) := \mathsf{DCT}(\mathbf{a})$, where $\mathbf{a} = (a_0,a_1, \dots,a_{N-1})^T$. Related to the cosine transform, we define a grid of points 
\begin{align}
    x_j := 2\cos\left(\frac{2\pi(2j + 1)}{4N}\right), \quad j = 0, 1,\hdots, N-1.
    \label{def:chebyshev_nodes}
\end{align}
The evaluation of $p(x)$ at the grid points $x_j$ yields
$$p(x_j) = a_0 + 2\sum_{i=1}^{N-1}a_i\cos\left( \frac{2\pi(2j+1)i}{4N} \right).$$
Moreover, we have the following property 
$$\mathbf{\hat{p}} = 2 \, \mathsf{DCT}(p(x)),$$
where $ \mathbf{\hat{p}} =(p(x_0),p(x_1),\dots,p(x_{N-1}))^T$ is the vector of all the evaluations.

Finally, let $p(x), g(x) \in \Z[x] $ and $r(x) = p(x)q(x)$. By definition the evaluations satisfy $r(x_j) = p(x_j)q(x_j)$ for all $x_j$ in the grid. Thus, the vector evaluations satisfy the property  
$$\mathbf{\hat{r}} = \mathbf{\hat{p}} \odot \mathbf{\hat{q}},$$ 
where $\odot$ denotes the componentwise product of vectors, that is,
\begin{equation}
\mathsf{DCT}(r(x)) = 2 \, \mathsf{DCT}(p(x)) \odot \mathsf{DCT}(q(x)).
\end{equation}
This is the approach that we shall use to derive a fast algorithm for multiplication in quotient rings of the form $\Z[x] / (f(x))$ for a special family of modulo polynomials $f(x)$.

\subsection{Case $n = 2^r$}

For $n=2^r$, the maximal real subextension of $\Q(\zeta_n)$ has dimension $m= \phi(n)/2 = 2^{r-2}$. In order to avoid trivialities, we consider only the case $r \geq 3$. 

It is well known that $\Psi_n(x) = V_m(x)$ \cite[Theorem~2.6]{loper2016resultants} where $V_m(x)$ is the modified Chebyshev polynomial from Definition \ref{def:modified_chebyshev_polynomials_V}. Every element in the ring $R:= \Z[x]/(\Psi_n(x))$ can be represented by a polynomial of degree less than or equal to $m-1$. Furthermore, if $p(x),q(x) \in R $ and $N=m$, then $\mathsf{DCT}(p(x))$ and $\mathsf{DCT}(q(x))$ can be computed in order $n \log n$ \cite{bi1998dct, bi1999fast}. 

In addition, since for $n=2^r$ the grid points $x_j$ for $j=0,1,\dots,N-1$ match the roots of $\Psi_n(x) = V_m(x)$, then for every polynomial $f(x) \in \Z[x]$ we have that $f(x_j) = \overline{f}(x_j)$, where $\overline{f}(x) \in R$ denotes the residue of $f(x)$ modulo $\Psi_n(x)$. This implies that if $r(x) = p(x)q(x)$, then $\mathbf{\hat{r}} = \hat{\mathbf{\overline{r}}}$ and
\begin{align*}
\mathsf{DCT}(\overline{r}(x)) &= \mathsf{DCT}(r(x)) \\
&= 2 \, \mathsf{DCT}(p(x)) \odot \mathsf{DCT}(q(x)).
\end{align*}
If $\overline{r}(x) = \sum\limits_{i=0}^{m-1}c_iV_i(x) \in R,$ then $\mathbf{c} = (c_0, c_1, \hdots, c_{m-1})^T$ can be computed using the inverse of the $\mathsf{DCT}$ and Lemma \ref{lemma dct}, that is,
\begin{equation}
\mathbf{c} = \frac{4}{N} \, \mathsf{IDCT}( \mathsf{DCT}(p(x)) \odot \mathsf{DCT}(q(x))).
\end{equation}

The fast implementation of the inverse transform $\mathsf{IDCT}$ requires in total $(r-2)2^{r-2}$ multiplications and $3(r-2)2^{r-3} - 2^{r-2} + 1$  additions \cite{kok1997fast, hou1987fast}. This means the asymptotic complexity of computing $\mathbf{c}$ is $\mathcal{O}(m \log m)$.

\subsection{Case $n=2^r 3^s$}

As before, to avoid trivialities we only consider the cases $r\geq 3$ and $s \geq 1$. For $n=2^r 3^s$, the maximal real subextension of $\Q(\zeta_n)$ has dimension $m = \phi(n) / 2 = 2^{r-1}3^{s-1}$, and the minimal polynomial is $\Psi_n(x) = V_m(x) - 1.$ This fact can be easily verified by computing 
\begin{align*}
V_m\left(2\cos\left(\frac{2\pi}{n} \right)\right) - 1 &= 2\cos\left(\frac{2\pi m}{n} \right) - 1 \\
&= 2\cos\left(\frac{2\pi}{6} \right) -1  \\
&= 0.
\end{align*}
Thus, $\psi_n$ is root of $V_m(x) - 1$ which is a monic polynomial of degree $m$.

To describe the fast multiplication, let $p(x), q(x) \in R = \Z[x]/(\Psi_n(x))$. Then $p(x)$ and $q(x)$ have degrees strictly less than $\deg \Psi_{n}(x) = m$. Therefore, the product $r(x)=p(x)q(x)$ has degree less than or equal to $N-2$ with the choice $N=2m$. We pick $N$ to be the size of the $\mathsf{DCT}$, we write
\begin{align*}
    r(x) = p(x)q(x) = \sum_{i=0}^{N-1} c_iV_i(x).
\end{align*}
Hence the coefficient vector $r(x)$ can be written as
\begin{align*}
    \mathbf{c} &= \frac{2}{N} \mathsf{IDCT}(\mathsf{DCT}(r(x))) \\
    & = \frac{2}{N} \mathsf{IDCT}(2\, \mathsf{DCT}(p(x)) \odot \mathsf{DCT}(q(x)))\\
    & = \frac{4}{N} \, \mathsf{IDCT}(\mathsf{DCT}(p(x)) \odot \mathsf{DCT}(q(x)))
\end{align*}

Both the $\mathsf{DCT}$ and the $\mathsf{IDCT}$ of size $N$ can be computed in $\mathcal{O}(N \log N)$ \cite{bi1999fast, kok1997fast}, even for the composite dimension $N=2^{r}3^{s-1}$. This means that we can compute $\mathbf{c}$ and hence $r(x) = p(x)q(x)$ in $\mathcal{O}(N \log N)$. 

Now, in order to compute the residue $\overline{r}(x) \in R$ of $r(x)$, we notice that $V_m(x) \equiv 1 \mod \Psi_n(x)$. By using the trigonometric identity 
\begin{align*}
    \cos(\alpha + \beta) + \cos(\alpha - \beta) = 2\cos \alpha \cos \beta,
\end{align*}
we obtain for degrees $i = 1, \hdots, m-1$ a formula for the reduction
\begin{align}
V_{m+i}(x) &= V_i(x)V_m(x) - V_{m-i}(x) \nonumber \\ 
&\equiv V_i(x) - V_{m-i}(x) \mod \Psi_n(x).
\label{reduction_mod_psi_n}
\end{align}
By using the above relations on $r(x) = p(x)q(x) = \sum\limits_{i=0}^{N-1}c_iV_i(x)$, we obtain
\begin{align*}
    \overline{r}(x) &= d_0 + d_1 V_1(x) + \hdots + d_{m-1}V_{m-1}(x),
\end{align*}
where the polynomial coefficients $d_i$ can be computed as a sums of the elements of the vector $\mathbf{c}$. The explicit formula is
\begin{align*}
    \mathbf{d}
    = 
    \begin{bmatrix*}
        c_0 + c_m \\
        c_1 + c_{m+1} - c_{2m-1} \\
        \vdots \\
        c_i + c_{m+i} - c_{2m-i} \\
        \vdots \\ 
        c_{m-1} - c_{m+1}
    \end{bmatrix*}.
\end{align*}
The number of operations needed for the reduction of $r(x)$ to $\overline{r}(x)$, meaning the transform from $\mathbf{c}$ to $\mathbf{d}$, is linear in $m$. Hence, the overall complexity is asymptotically dominated by the $\mathcal{O}(N \log N)$ complexity of the $\mathsf{DCT}$ and $\mathsf{IDCT}$. As a conclusion, the overall complexity of computing a product of two elements in $R$ is $\mathcal{O}(m \log m)$.

\subsection{DCT modulo a prime number}

As usual in cryptography and other contexts, in particular in LWE schemes, it is necessary to reduce integers modulo a prime number. The algorithms developed earlier might be difficult to implement numerically due to the nature of values of the cosine function. Nevertheless, this difficulty disappears when working in the quotient ring $R_q = \Z_q[x]/(\Psi_n(x))$ for a suitable choice of the prime $q$.    

For $n=2^r$, let $M = n$. Then it is enough to choose a prime $q$ satisfying $q \equiv 1 \mod M$, because this implies the existence of a primitive $M$-th root of unity, $\zeta_M$, in the finite field $\F_q = \Z/q\Z$. 
For $n = 2^r 3^s$, we choose a prime $q$ such that $q \equiv 1 \mod M$, where $M = 4N =2^{r+2}3^{r-1}$. The value of $M$ is consequence of our choice of size $N$ for the $\mathsf{DCT}$ earlier in Section \ref{sec:fast_multiplication_via_DCT}.

In both cases, for all indices $i$ and $j$, we use 
$$ 2\cos\left( \frac{2\pi(2j+1)i}{4N}\right) = \zeta_M^{(2j+1)i} + \zeta_M^{-(2j+1)i} \in \F_q.$$
Likewise, all the elements of the matrices $\mathbf{C}_N$ and $\mathbf{C}_N^{-1}$ can be seen as elements of $\F_q$, and in particular the $\mathsf{DCT}$ and $\mathsf{IDCT}$ can be computed over the elements of $\F_q$.

In fact, our multiplication algorithm extends naturally to a composite modulus $pq$ for two primes $p$ and $q$ satisfying the equivalence conditions above. This follows as a direct application of the Chinese Remainder Theorem. We use the ring isomorphism $\Z / pq \Z \cong \F_p \times \F_q$ to construct the required root of unity in $\Z / pq \Z$ that maintains its exact order under the projections to $\F_p$ and $\F_q$. This is enough to define the $\mathsf{DCT}$ and $\mathsf{IDCT}$ over $\Z / pq \Z$. Note that the existence of the inverse of $2$ and $N$ is also guaranteed by the isomorphism.

\section{PLWE in $\Z[x] / (\Psi_n(x))$ via fast base change}
\label{sec:PLWE_in_maxreal}
Let $n = 2^r$ or $n = 2^r 3^s$ with the corresponding $m = \phi(n) / 2$, and consider the PLWE problem in the quotient ring $R = \Z[x] / (\Psi_n(x))$. Under the PLWE--RLWE equivalence shown in Section \ref{sec:RLWE-PLWE_equivalence}, we want to keep the security guarantees of the PLWE scheme which is typically stated in the power basis. This means that we must sample the polynomials $a(x)$ and $s(x)$ in the power basis with coefficient vectors drawn uniformly on $\mathbb{F}_q^m$.

To obtain fast multiplication of $a(x)$ and $s(x)$ in the ring $R$, the idea is to first perform a change of basis and compute the product $a(x) \cdot s(x)$ in the modified Chebyshev basis $\{V_0(x), V_1(x), \hdots, V_{m-1}(x)\}$ using the results from Section \ref{sec:fast_multiplication_via_DCT}. After the fast computations in the modified Chebyshev basis, we perform another change of basis back to the power basis giving us the product $a(x)s(x)$ in the original basis.

In fact, all of these operations can be computed using a total of $\mathcal{O}(n \log n)$ arithmetic operations. The quasilinear complexity of the multiplication in the modified Chebyshev basis was already proved in Section \ref{sec:fast_multiplication_via_DCT}. This section focuses on proving that the two change of basis transforms can also be done in $\mathcal{O}(n \log n)$. This result is stated as the following lemma.
\begin{lemma}
Given a polynomial of degree less than or equal to $m-1$, the complexity of the change of basis between the power basis $\{1, x ,x^2, \dots, x^{m-1} \}$ and $\{ V_0(x), V_1(x), \dots$ $ ,V_{m-1}(x)\}$ is $\mathcal{O}(m \log m)$.
\label{lemma:fast_base_change}
\end{lemma}

Before we can prove Lemma \ref{lemma:fast_base_change}, we need a few more results regarding the complexity of polynomial evaluation on the Chebyshev nodes $x_j$ \eqref{def:chebyshev_nodes} and Chebyshev interpolation. In \cite{pan1998new, gerasoulis1988fast}, the authors describe an algorithm to evaluate a polynomial $p(x)$ of degree at most $N-1$ at the Chebyshev nodes $x_j$  with complexity $\mathcal{O}(N \log N)$. In \cite{pan1998new}, this result is stated as the following theorem.
\begin{theorem}
Any polynomial of degree at most $N-1$ can be evaluated on the Chebyshev
nodes $x_j$, $j = 0, \hdots, N-1$, at the cost of $\mathcal{O}(N \log N)$ arithmetic operations.
\label{thm:chebyshev_evaluation}
\end{theorem}
The authors also provide the following result on the efficiency of interpolation on the Chebyshev nodes.
\begin{theorem}
Interpolation to a polynomial of a degree at most $N-1$ on the Chebyshev node set can be performed in $\mathcal{O}(N \log N)$ arithmetic operations.
\label{thm:chebyshev_interpolation}
\end{theorem}
We are now ready to prove Lemma \ref{lemma:fast_base_change}

\begin{proof}
First, let us explore the complexity of the change from the power basis $\{1,x,x^2,$ $ \dots, x^{m-1} \}$ to the modified Chebyshev basis
\begin{align*}
    V:= \{V_0(x), V_1(x),\dots, V_{m-1}(x) \}.
\end{align*}
Let $p(x) = \sum_{i=0}^{m-1} p_i x^i$ be a polynomial in $\Z[x]/(\Psi_n(x))$ and consider its expansion in the $V$-basis as
\begin{align*}
    p(x) &= \sum_{i=0}^{m-1} c_i V_i(x).
\end{align*}
For the change of basis, given the coefficients $p_i$ our goal is to find the coefficients $c_i$.

Let $x_j$ denote the Chebyshev nodes from \eqref{def:chebyshev_nodes}.
As was shown in Section \ref{sec:fast_multiplication_via_DCT}, the vector of evaluations $\mathbf{\hat{p}} = \left(p(x_0), p(x_1), \dots, p(x_{N-1}) \right)$ and the vector of coefficient $\mathbf{c} = \left(c_0, c_1, \dots, c_m, 0, \dots, 0 \right)$ satisfy the relations
\begin{align}
\mathbf{\hat{p}} &= 2 \,\mathsf{DCT}( \mathbf{c}) \label{eq:phat_via_DCT}\\
\mathbf{c} &= \frac{N}{4} \mathsf{IDCT}(\mathbf{\hat{p}}). \label{eq:c_via_IDCT}
\end{align}
In other words, given the vector of evaluations $\mathbf{\hat{p}}$, the vector of coefficients $\mathbf{c}$ can be computed in order $N\log(N)$. This is formalized by Equation \eqref{eq:c_via_IDCT}. Vice-versa, given the vector of coefficients $\mathbf{c}$, we can compute the vector of evaluations via a single $\mathsf{DCT}$ with complexity $\mathcal{O}(N\log(N))$.

Theorem \ref{thm:chebyshev_evaluation} states that given a polynomial 
$$p(x) = \sum_{i=0}^{m-1} p_i x^i $$
in the power basis, the vector of evaluations $\mathbf{\hat{p}}$ can be computed in $\mathcal{O}(N \log N)$, and from this, we apply the respective $\mathsf{DCT}$ to retrieve the coefficient vector $\mathbf{c}$ in base $\{V_0(x), V_1(x), \dots, V_{m-1}(x)\}$. Therefore, the total order of the base change is $\mathcal{O}(N \log N)$.

Next, we consider the change of basis from $\{V_0(x), V_1(x), \dots, V_m(x)\}$ to the power basis. Let $q(x) \in \Z[x]/(\Psi_n(x))$ be a polynomial and write its representation in the $V$-basis as 
$$q(x) = \sum_{i=0}^{m-1} a_i V_i(x).$$

By using using \eqref{eq:phat_via_DCT}, the vector of evaluations at the Chebyshev nodes, can be computed with $\mathcal{O}(N \log N)$ arithmetic operations. Note that this statement is similar to that of Theorem \ref{thm:chebyshev_evaluation}, but here the polynomial to be evaluated is given in the $V$-basis. Now, the interpolation problem of finding the coefficients $q_i$ of $ q(x)$ in the power basis, i.e, the form
$$q(x) = \sum_{i=0}^{m-1} q_i x^i,$$
can be solved in order $\mathcal{O}(N \log N)$ due to the interpolation Theorem \ref{thm:chebyshev_interpolation}.
\end{proof}

We are ready to state our second main result.
\begin{theorem}
Given two polynomials $a(x), s(x) \in \Z[x] / (\Psi_n(x))$ in the power basis, their product $a(x) \cdot s(x) \in \Z[x] / (\Psi_n(x))$ can be computed with asymptotic complexity $\mathcal{O}(n \log n)$.
\end{theorem}
\begin{proof}
    With Lemma \ref{lemma:fast_base_change} and the fast multiplication algorithm from Section \ref{sec:fast_multiplication_via_DCT}, the proof is evident.
\end{proof}

\section{Small roots of cyclotomic polynomials modulo a prime}
\label{sec:root_computations}

Let $\mathcal{O}_f = \Z[x] / (f(x))$ and $R_q = \mathcal{O}_f / q\mathcal{O}_f$ as before. The root-based attacks from \cite{ELOS:2016:RCN} show that if the polynomial $f(x)$ has a small root modulo the reducing prime $q$, then there exists a polynomial time evaluation attack against the decisional PLWE problem. In this section, we compare heuristically cyclotomic fields and their maximal real subfields for attacks against roots of small residue. The numerical computations are based on the paper \cite{Blanco:2021:VMTR}, where the PLWE--RLWE equivalence is proved for conductors of the form $4p$, where $p$ is a prime. For a reasonable range of parameters, our goal is to show that on average the minimal polynomials of the maximal real extensions are not worse than cyclotomic polynomials when it comes to having small roots. In practice, this means that we can easily find a pair of primes $(p,q)$ such that the polynomial $\Psi_{4p}(x)$ with prime modulus $q$ is not susceptible to attacks against roots of small residue.

We show that between the two families of minimal polynomials, $\Phi_p(x)$ and $\Psi_{4p}(x)$ of the same degree, the maximal real extensions $\Psi_{4p}(x)$ tend to have fewer small roots $x \in S = \{\pm 2, \pm 3, \pm 4, \pm 8\}$ when $p$ is large enough. On the other hand, for cyclotomic polynomials of small degree $p < 200$ and reducing primes $q < 5 \cdot 10^{10}$, we did not encounter any small roots in $S$. 

The roots in the set $S$ are small powers of two. We also included another small prime 3 because our preliminary computations suggested that it might occur as a root relatively often. This choice allowed us to study the frequency of the roots in $S$ to see if it varies a lot between the elements of $S$.

For the prime $q$ and $n = \deg f(x)$, the parameters of the three final lattice-based NIST candidates are:
\begin{itemize}
\item $q=3329$ and $n=256$ for ML-KEM, formerly known as Kyber \cite{FIPS:2024:203}.
\item $q=8380417$ and $n=256$ for ML-DSA, formerly known as Dilithium \cite{FIPS:2024:204}.
\item $q=12289$ and $n=512$ in security level 1 and $n=1024$ in security level of Falcon \cite{Falcon}.
\end{itemize}
The ML-KEM and ML-DSA are the first two lattice-based PQC standards.

In practice, for PLWE instances the sparsity of the vulnerable polynomials allows us to pre-emptively test for small zeros for different prime moduli $q$ and choose a safe one. In this regard, the weak polynomials we present are only of theoretical interest in the PQC realm. However, it is worth mentioning that in \cite{BDM:2024:GARB}, the authors attack the PLWE problem by using as the modulus the NTRU polynomials with parameters:
\begin{itemize}
    \item $n = 509$, $q = 2048$, $\sigma = 8$.
    \item $n = 677$, $q = 2048$, $\sigma = 8$.
    \item $n = 821$, $q = 4096$, $\sigma = 8$.
\end{itemize}
This attack is not directly based on roots but, on distinguishing the PLWE distribution from the uniform distribution with non-negligible advantage by using a refined statistical test. However, like the root-based attacks, this attack takes advantage of the evaluation of the polynomials at special points. In any case, it does not imply that an attack the NTRU cryptosystem exists, as this is backed by a different security consideration, i.e. the NTRU problem itself. Regardless, this is evidence that even the real-world parameters can be affected by novel families of attacks.

The situation is slightly different in the homomorphic encryption setting, where much larger values of both the degree $n$ and the prime modulus $q$ can occur. For instance, within HELib \cite{gouert}, a large initial ciphertext modulus of 780 bits is used, which needs a cyclotomic degree of 55831. Thus, the ring is of dimension $\phi(55831)=54000$ achieving 129 bits of security. More specifically, in \cite[Table 4.2]{albertoreciente} we can observe that $n$ runs in powers of 2 from 1024 to 131072, and $q$, depending on $n$, runs between 29 to 1920 bits. These parameters refer to power-of-two cyclotomic polynomials.


In \cite{BL:2022:RPE,Blanco:2021:VMTR}, the authors show that for the maximal real cyclotomic extensions, the minimal polynomials $\Psi_n(x)$ do not have roots at $x = \pm 2$ for any reducing odd prime $q$ when $n = 2^r k$ for an odd $k$. Moreover, $x = \pm 1$ are never roots modulo an odd prime $q$. As a result, the family of polynomials $\Psi_n(x)$ is immune against the attacks against roots of small order with $x = -1$ and attacks against roots of small residue with $x = \pm 2$. This is not the case for cyclotomic polynomials, since the authors provide an example of a cyclotomic polynomial $\Phi_{61}(x)$ that does have a root at $x = 2$ modulo $q = 9520972806333758431$. Again, we emphasize that this is a theoretical threat due to the size of the modulus $q \approx 10^{19}$ being astronomical and thus infeasible in practice.


For primes from $p < 1500$, we computed the evaluations of $\Phi_p(x)$ and $\Psi_{4p}(x)$ at the elements of the set $S$ and then reduced the evaluations with a varying prime modulus $q$ running as high as $5 \cdot 10^{10}$. Note that thank to the result from \cite{BL:2022:RPE,Blanco:2021:VMTR}, we could exclude the possibility of $\alpha = \pm 2$ being a root of $\Psi_{4p}(x)$ modulo any prime $q$. Furthermore, the minimal polynomials $\Psi_{4p}(x)$ are even \cite{loper2016resultants} so $\alpha \in S$ is root if and only if $-\alpha \in S$ is a root. Thus, in our numerical search the number of roots at $\alpha$ and $-\alpha$ are the same for $\Psi_{4p}(x)$. Finally, note that for a fixed $p$, the degrees of the two polynomials are equal 
\begin{align*}
    \deg \Phi_p(x) &= \phi(p) = p - 1, \\
    \deg \Psi_{4p}(x) &= \phi(4p) / 2 = \phi(p) = p-1.
\end{align*}

The visualization of the computational results is given in Figure \ref{fig:cycloroots_and_maxrealroots} (a) and (b), respectively for $\Phi_p(x)$ and $\Psi_{4p}(x)$. If a polynomial has more than one small root in $S$, the smallest root in absolute value is plotted. 

\newpage
\thispagestyle{empty}
\begin{figure}[H]
 \centering
 \begin{subfigure}[b]{\textwidth}
     \centering
     \includegraphics[width=0.98\textwidth, trim={0 7.75cm 0 7cm}, clip]{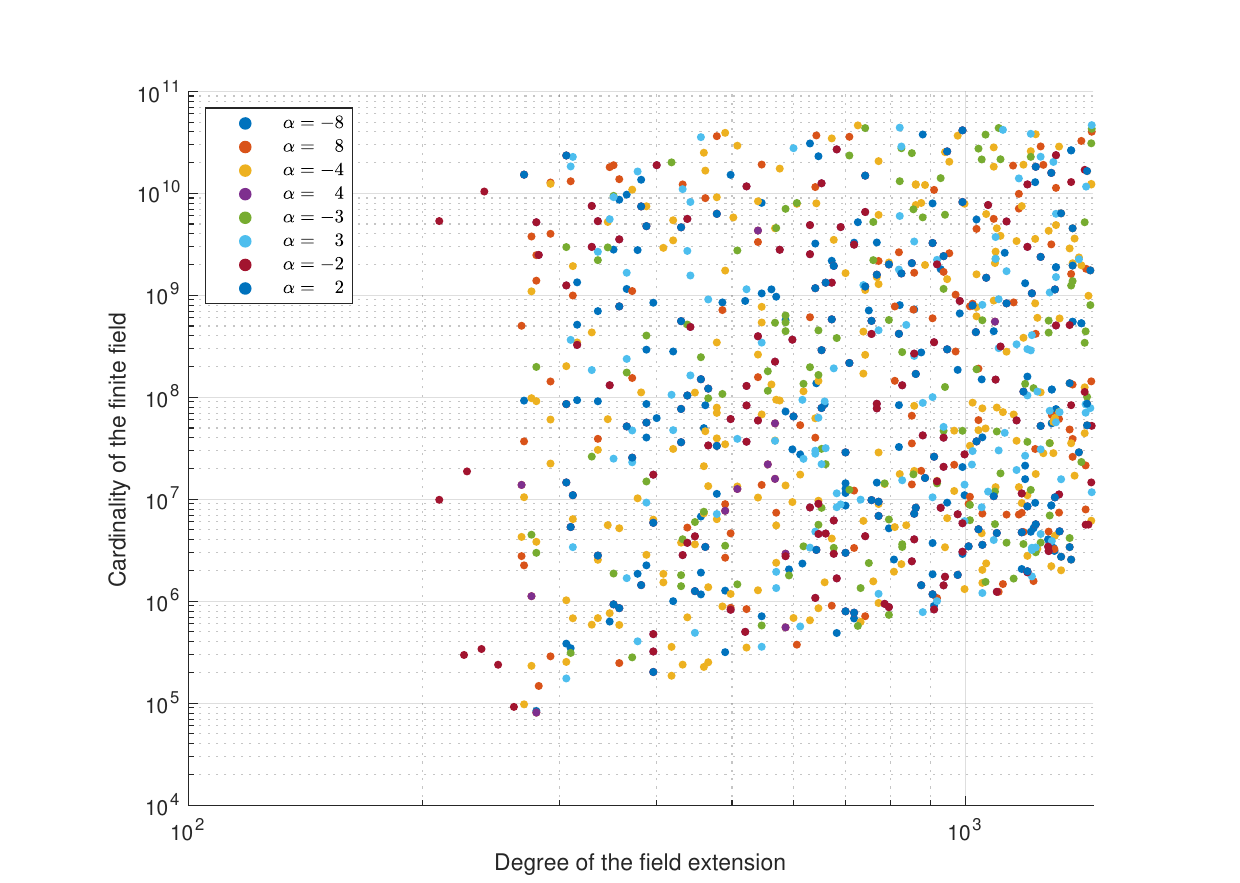}
     \caption{Small roots of the cyclotomic polynomials $\Phi_p(x)$.}
 \end{subfigure}
 \begin{subfigure}[b]{\textwidth}
     \centering
     \includegraphics[width=0.98\textwidth ,trim={0 7.75cm 0 7cm}, clip]{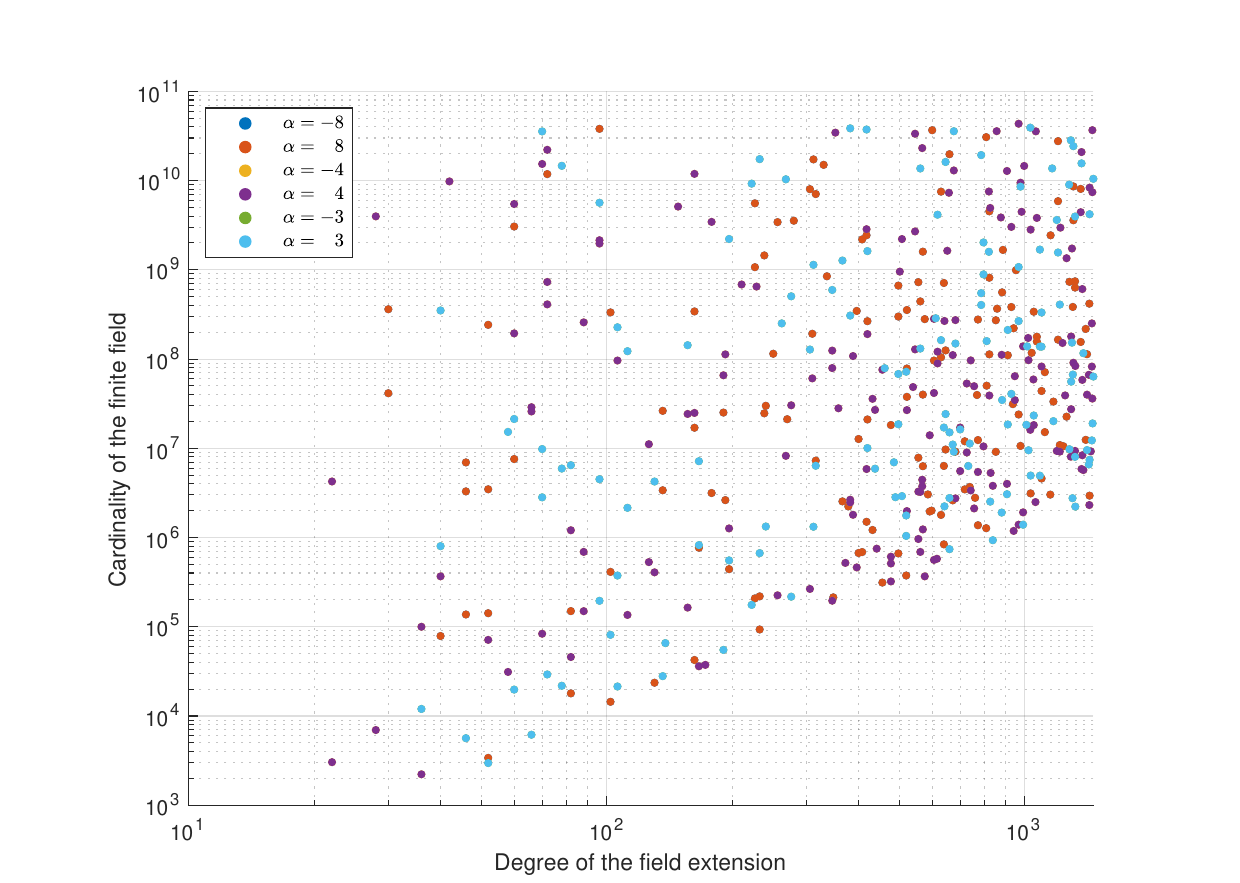}
     \caption{Small roots of the maximal real cyclotomic polynomials $\Psi_{4p}(x)$.}
 \end{subfigure}
    \caption{Each point of the plot is a pair $(p-1,q)$ for primes $p$ and $q$ such that $f_p(\alpha) \equiv 0 \mod q$ for some $\alpha \in S$ where $f_p(x) = \Phi_p(x)$ or $f_p(x) = \Psi_{4p}(x)$. For both polynomials the degree of the field extension is $\deg f_p(x) = p-1$ and the cardinality of the finite field is $q$.}
    \label{fig:cycloroots_and_maxrealroots}
\end{figure}

For primes $p$ less than $1500$ and $q$ less than $5 \cdot 10^{10}$, we found 1191 small roots in $S$ for the cyclotomic polynomials $\Phi_p(x)$. For the same range of primes, the maximal real cyclotomic polynomials had fewer small roots in $S$, only 926. The distribution of roots of the two polynomial families with no indication of the precise value of $\alpha$ is given in Figure \ref{fig:cyclo_vs_maxreal}. The small zeroes of the cyclotomic polynomials $\Phi_p(x)$ seem to appear only at higher degree polynomials with $p > 200$, whereas the maximal real cyclotomic polynomials start having small zeroes already at $p = 23$.

The counts for the number of zeros we found for the evaluation points in $S$ is given in Table \ref{tab:number_of_zeros_found_at_x}. We did not observe a significant difference between the number of zeros at the small powers of 2 or at the small prime $3$. In Table \ref{tab:number_of_zeros_found_at_x} the number of zeros of $\Psi_{4p}(x)$ at $\alpha$ and $-\alpha$ are equal due to the even parity of the polynomials.

In summary, the small zeroes of $\Phi_p(x)$ are more frequent in the primes $p > 200$, that is, the density of the small zeros of $\Phi_p(x)$ appears to be larger as the red points in Figure \ref{fig:cyclo_vs_maxreal} are more concentrated to the right where $p > 200$. On the other hand, $\Psi_{4p}(x)$ has overall less small zeros in the degree range $p < 1500$ and $q < 5 \cdot 10^{10}$, and the zeros are more sparse and evenly distributed in the $p$--$q$-plane. The number of zeros we found does not vary significantly between the elements of $S$ with counts ranging approximately between 100 to 200.
\begin{figure}[H]
    \begin{center}
        \includegraphics[width=\textwidth, trim={0 7.75cm 0 7cm}, clip]{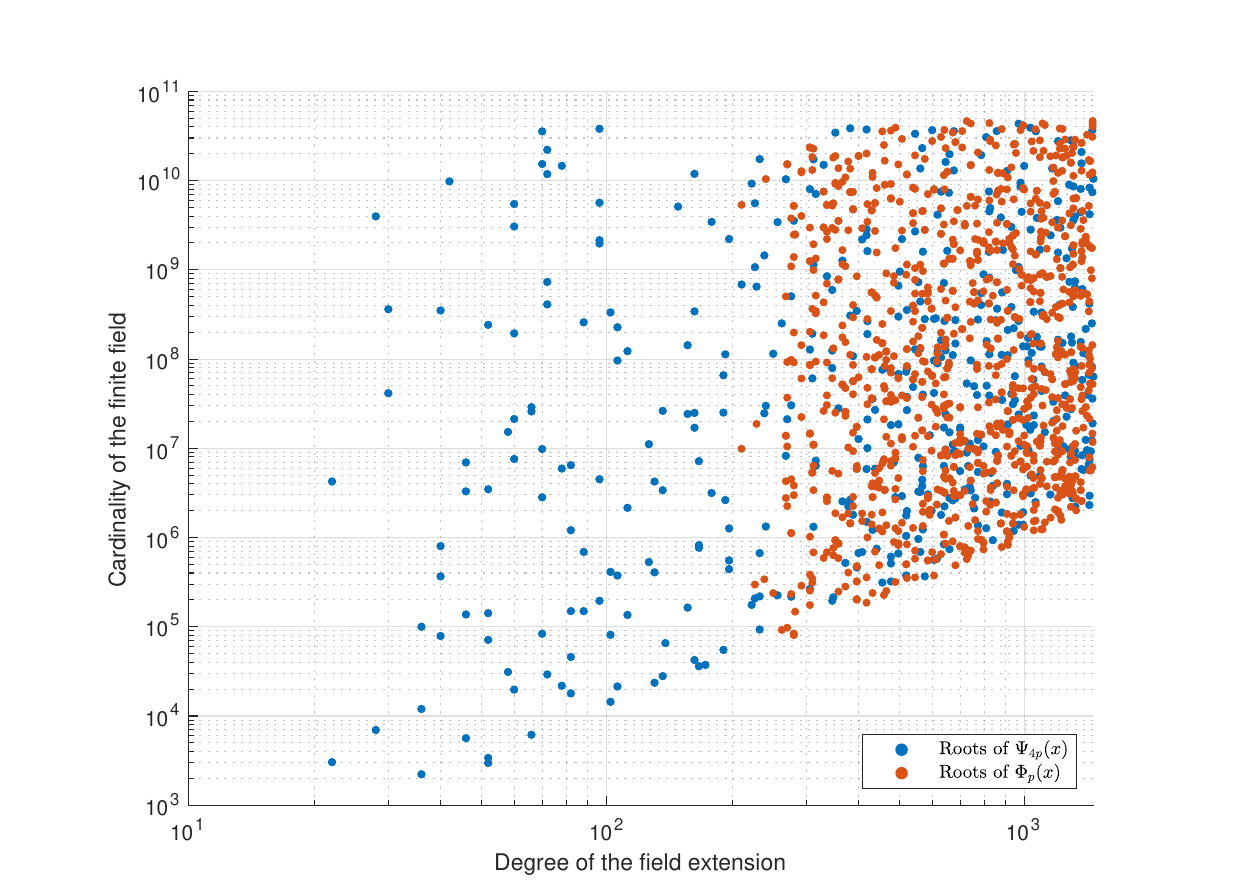}
    \end{center}
    \caption{Small roots $\alpha \in S$ for the pairs $(p,q)$. The points are coloured in red for a root of $\Phi_p(x)$ and in blue for $\Psi_{4p}(x)$.}
    \label{fig:cyclo_vs_maxreal}
\end{figure}

\begin{table}
\caption{Number of zeros of $\Phi_p(x)$ and $\Psi_{4p}(x)$ for elements in $S$ with $p < 1500$ and  $q < 5 \cdot 10^{10}$.}
\label{tab:number_of_zeros_found_at_x}
\begin{tabular}{c|cccccccc}
                & $x = -2$ & $x = 2$ & $x = -3$ & $x = 3$ & $x = -4$ & $x = 4$ & $x = -8$ & $x = 8$ \\\hline
$\Phi_p(x)$     & 105      & 102     & 117      & 121     & 213      & 190     & 121      & 222     \\
$\Psi_{4p}(x)$  & 0        & 0       & 140      & 140     & 172      & 172     & 151      & 151   
\end{tabular}
\end{table}

\bibliography{DCC_fastmult}


\end{document}